\documentclass[12pt]{article}
\pdfoutput=1
\usepackage{amsmath,amssymb, braket, amsthm, bm}
\usepackage[a4paper]{geometry}
\usepackage[colorlinks=true, pdfstartview=FitV, linkcolor=blue, citecolor=blue, urlcolor=blue]{hyperref}
\usepackage{graphicx}
\usepackage{subcaption}
\usepackage{xcolor}

\newcommand{\ffi}{\varphi}

\def\tr{\mathop{\rm Tr}\nolimits}
\newcommand{\ind}{\mbox{ index}}
\def\ran{\mathop{\rm Ran}\nolimits}

\newcommand {\bC}{{\mathbb C}}
\newcommand {\bZ}{{\mathbb Z}}
\newcommand {\bN}{{\mathbb N}}

\newcommand {\bR}{{\mathbb R}}
\newcommand {\bS}{{\mathbb S}}
\newcommand {\bH}{{\mathbb H}}
\newcommand {\bI}{{\mathbb I}}
\newcommand {\hil}{{\mathcal H}}
\newcommand {\cL}{{\mathcal L}}
\newcommand {\cF}{{\mathcal F}}
\newcommand {\cG}{{\mathcal G}}
\newcommand {\cS}{{\mathcal S}}
\newcommand {\bbS}{{\bf S}}
\newcommand {\flux}{{\Phi}}
\newcommand{\evenunder}[1] {\bm{#1}}
\newcommand{\evenover}[1] {\bm{#1}}

 \newcommand {\fatpsi}{\bm{\psi}}
 \newcommand{\fatU}{\bm{U}}
  \newcommand{\fatchi}{\bm{\chi}}
  \newcommand{\fatS}{\bm{\cS}}

\newcommand{\be}{\begin{equation}}
\newcommand{\ee}{\end{equation}}

\newtheorem{thm}{Theorem} [section]
\newtheorem{lem}[thm]{Lemma}
\newtheorem{prop}[thm]{Proposition}

\newtheorem{definition}[thm]{Definition}
\newtheorem{cor}[thm]{Corollary}

\newtheorem {rem}[thm]{Remark}
\newtheorem {rems}[thm]{Remarks}

\title{ Chirality induced Interface Currents in the Chalker Coddington Model
\thanks{ Supported by
FONDECYT 1161732, 
and ECOS-Conicyt C15E10}\  \thanks{Supported by the LabEx PERSYVAL-Lab (ANR-11- LABEX- 0025-01) funded by the French program Investissement d'Avenir}}
\author{Joachim Asch  \thanks{CNRS, CPT, Aix Marseille Universit\'e, Universit\'e de Toulon, Marseille, France, asch@cpt.univ-mrs.fr},
Olivier Bourget
\thanks{
Departamento de Matem\'aticas
Pontificia Universidad Cat\'olica de Chile, Av. Vicu\~{n}a Mackenna 4860,
C.P. 690 44 11, Macul
Santiago, Chile},
Alain Joye
\thanks{
Universit\'e Grenoble Alpes, CNRS Institut Fourier, 38000 Grenoble, France}
}
\date{3/5/18}
\begin{document}
\maketitle

\begin{abstract}

We study transport properties of a Chalker--Coddington type model in the plane which presents asymptotically pure anti-clockwise rotation on the left and clockwise rotation on the right. We prove delocalisation in the sense that the absolutely  continuous spectrum  covers the whole unit circle. The result is of topological nature and independent of the details of the model.
\end{abstract}

\section{Introduction}

By a Chalker--Coddington (aka: CC--model) we understand a unitary operator
\[U_{CC}:\ell^2(\bZ^2;\bC)\to\ell^2(\bZ^2;\bC)\]
defined by a collection of $2\times2$ scattering matrices, i.e.: a map
\begin{equation}\bbS:\bZ^2\to U(2), \qquad (j,k)\mapsto S_{j,2k}\in U(2).\label{eq:S}\end{equation}

Denoting by  $\ket{j,k}$ the canonical basis vectors of $l^2(\bZ^2;\bC)$, $\bbS$ defines $U_{CC}$  according to figure \ref{fig:scatteringnetwork} by:

\begin{align}
&\begin{pmatrix} U_{CC}|2j,2k\rangle \cr U_{CC}|2j+1,2k-1\rangle \end{pmatrix} := S_{2j,2k} \begin{pmatrix} |2j,2k-1\rangle \cr |2j+1,2k\rangle \end{pmatrix},\notag \\
&\begin{pmatrix} U_{CC}|2j+1,2k\rangle \cr U_{CC} |2j+2,2k+1\rangle \end{pmatrix} := S_{2j+1,2k} \begin{pmatrix} |2j+2,2k\rangle \cr |2j+1,2k+1\rangle \end{pmatrix}. \label{def:UCC}
\end{align}
\begin{figure}[hbt]
\centerline {
\includegraphics[width=8cm]{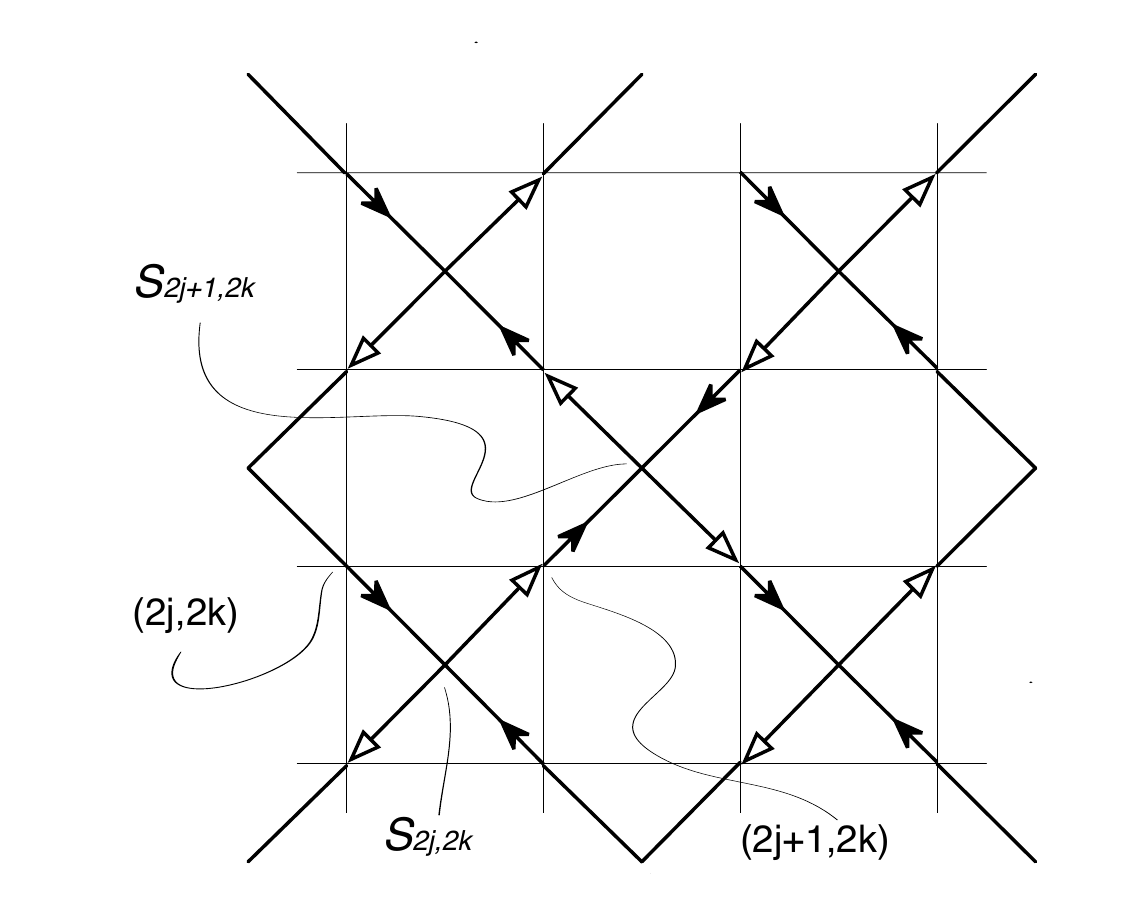}
}
\caption{A Chalker--Coddington model with its incoming (solid arrows) and outgoing links}
\label{fig:scatteringnetwork}
\end{figure}
According to the parity of the first index we will speak of odd and even scattering matrices $S_{j,2k}$.

The model provides an effective description of one time step of  the motion of an electron in a plane subject to a strong perpendicular magnetic field and electric potential whose main physical characteristics are encoded by the scattering matrices.  For details on its  known physical and mathematical features, we refer to  \cite{cc, kok, ABJ1, ABJ2};  here we  just mention that the great interest of this effective model stems from its ability to describe the delocalisation transition of the Quantum Hall effect.

If all scattering matrices $S_{j,2k}$ are off--diagonal then the motion is an  anti-clockwise rotation on four dimensional subspaces (aka: plaquettes); if they are all diagonal the motion is clockwise on (different) plaquettes. It is known that certain random perturbations of these cases display dynamical localisation  \cite{ABJ2}. The critical case, in the sense of stable delocalisation, is supposed to occur when all matrix entries are of modulus $1/\sqrt{2}$,  \cite{cc}.  On the other hand it is known that the translation invariant critical case has trivial Chern numbers \cite{F}.  In the present contribution, we investigate the spectral and transport properties on an interface made of arbitrary scattering matrices between two phases of different chirality. In particular we suppose that the motion is anti-clockwise in a left half-plane and clockwise in a right half-plane as depicted in figure \ref{fig:strip}.

\begin{figure}[hbt]
\centerline {
\includegraphics[height=6cm]{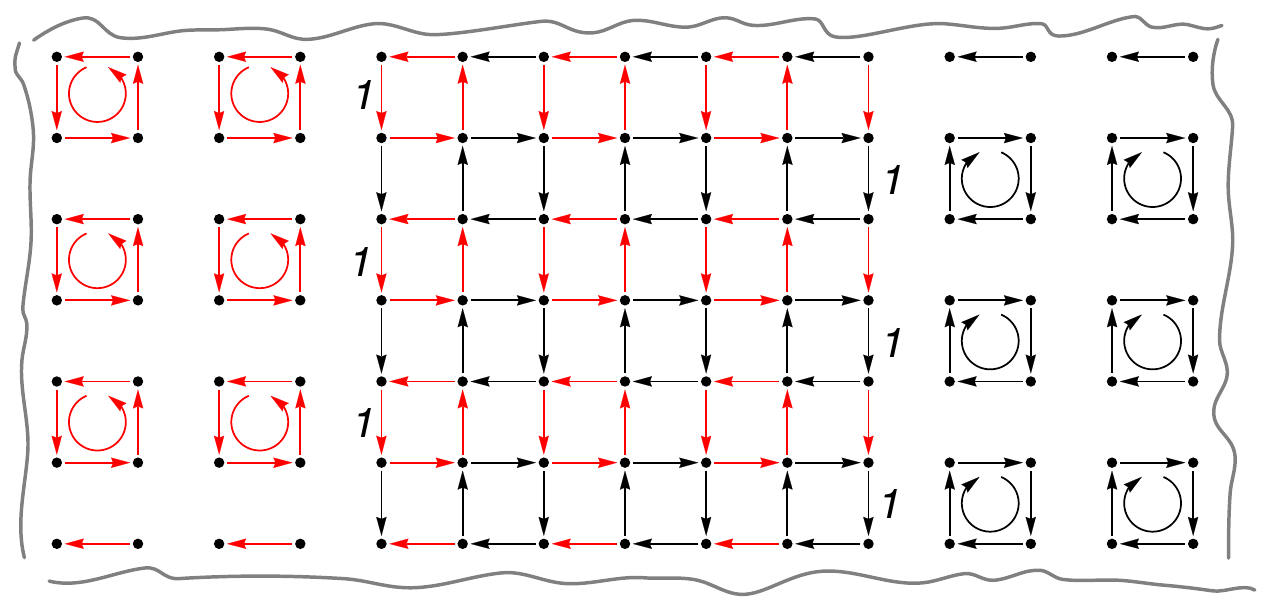}
}
\caption{The invariant strip; the indicated transition {probabilities} $1$ on the boundaries are imposed by the different chiralities}
\label{fig:strip}
\end{figure}

We  observe that the clockwise and anti-clockwise phases induce boundary conditions on the interface making it an  invariant strip under $U_{CC}$. Then we consider a natural flux observable in configuration space and prove that the spectrum of $U_{CC}$ restricted to the interface contains a non-trivial absolutely continuous component that covers the whole circle. This is independent of the details inside the interface. Our argument makes use of a topological quantity which comes as the index of a pair of projections. This has been used before for projections in energy space, in particular for the proof that the quantum Hall conductance is an index \cite{ASS,ASS1} and for the equivalence of Bulk and Edge Conductance \cite{elbaugraf,grafPorta}. The present point of view is in configuration space and inspired by \cite{ki}. Implications concerning the absolutely continuous spectrum seem to be new. 
For the sake of comparison, we recall that CC models with certain types of random scattering matrices, restricted to a strip with periodic boundary conditions,  display dynamical localisation  \cite{ABJ1}.  

{Bulk-Edge correspondence for unitary network models from a Floquet point of view  is discussed in \cite{kbrd, rlbl, dft, graftauber, ssb}. 
An ideologically similar situation occurs in the Iwatsuka model, see \cite{iwatsuka,mmp,dgr}.}

In the special case where the CC model is invariant under translations parallel to the interface, we show that the continuous spectrum of its restriction to the interface is purely absolutely continuous.  On the other hand, showing the presence of absolutely continuous spectrum by Mourre's method requires more information on the model \cite{ABJ3}.
We remark that the present results can adapted for two-dimensional coined  Quantum Walks (see {\it e.g.} \cite{ABJ3}) which we will do elsewhere.

\section{Properties of the model}

Next we discuss some basic properties of the model: its symmetry under parity, the special cases of (anti-)clockwise motion and the continuous dependence on the defining scattering matrices and the occurence of an invariant strip interface.

\begin{lem}
i) For any collection $\{S_{j,2k}\in U(2)\ \mbox{s.t. } (j,k)\in\bZ^2\}$,
 $U_{CC}$ is unitary.

ii) Let $\mathcal I$ be the involutive unitary parity operator  defined by  ${\mathcal I} |j,k\rangle  =(-1)^{j+k}|j,k\rangle$. Then 
$${\mathcal I}U_{CC}{\mathcal I}=-U_{CC} \ \Rightarrow \ \sigma(U_{CC})=-\sigma(U_{CC}).$$

iii) The adjoint $U^*_{CC}$ reads
\begin{align*}
&\begin{pmatrix}U_{CC}^* |2j,2k-1\rangle \cr U_{CC}^* |2j+1,2k\rangle \end{pmatrix} = S_{2j,2k}^* \begin{pmatrix} |2j,2k\rangle \cr |2j+1,2k-1\rangle \end{pmatrix}, \\
&\begin{pmatrix}U_{CC}^* |2j+2,2k\rangle \cr U_{CC}^* |2j+1,2k+1\rangle \end{pmatrix} = S_{2j+1,2k}^* \begin{pmatrix} |2j+1,2k\rangle \cr |2j+2,2k+1\rangle \end{pmatrix}. 
\end{align*}

\end{lem}

Denote the odd n--sphere by
\[\bS^n:=\left\{z\in\bC^{\frac{n+1}{2}}; \sum_{j} \vert z_j\vert^2 =1\right\}.\]
To parametrize the scattering matrices  $S_{j,2k}$, we use the homeomorphism
\begin{align}
S: \bS^1\times\bS^3\to U(2), \quad (q,(r,t))\mapsto  q\begin{pmatrix} r & -t \cr \overline{t} & \overline{r}\end{pmatrix}. \label{eq:su2}
\end{align}
So $\bbS$ is parametrised by a function $\bZ\times2\bZ\mapsto\bS^1\times\bS^3$, cf. : figure \ref{fig:network}.

\begin{figure}[hbt]
\centerline {
\includegraphics[width=8cm]{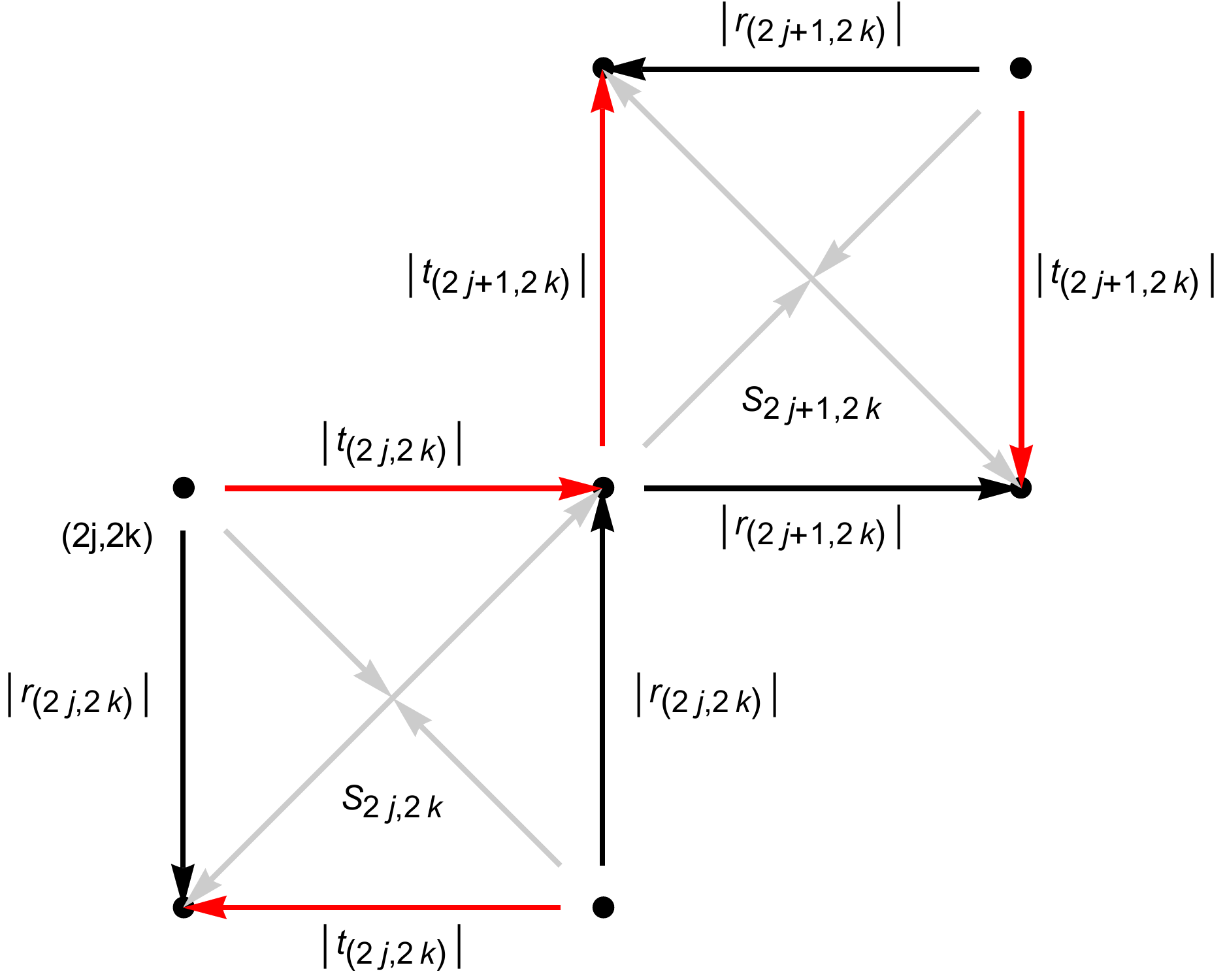}
}
\caption{The labeling of the matrix-elements of $S_{j,2k}$ }
\label{fig:network}
\end{figure}

One gets  the following characterisations of the right and left turning  phases, as well as the corresponding definition of plaquettes.  
\begin{lem} \label{LR}
In case all scattering matrices are diagonal, i.e. $t_{j,2k}=0 \Leftrightarrow  |r_{j,2k}|=1 $, the subspaces
$$\bH_{\circlearrowright}^{j,k}:= \mbox{span}\{|2j,2k\rangle, |2j,2k-1\rangle, |2j-1,2k-1\rangle, |2j-1,2k\rangle \}$$ 
are invariant under $U_{CC}$. The dynamics on those plaquettes is that of right  turners  with representation in the corresponding basis 
$$
U_{CC}|_{\bH_{\circlearrowright}^{j,k}}=\begin{pmatrix} 
0 & 0 & 0 & q_{2j-1, 2k}r_{2j-1, 2k} \cr
q_{2j, 2k}r_{2j,2k} & 0 & 0 & 0 \cr
0 & q_{2j-1, 2k-2}\overline{r_{2j-1,2k-2}}  & 0 & 0 \cr
0 & 0 & q_{2j-2, 2k}\overline{r_{2j-2,2k}}  & 0
\end{pmatrix}
$$
and spectrum $e^{\alpha^R_{2k,2j}}\{1,i, -1, -i\}$, where 
$$e^{4\alpha^R_{2k,2j}}=q_{2j, 2k}r_{2j,2k} q_{2j-1, 2k-2}\overline{r_{2j-1,2k-2}}  q_{2j-2, 2k}\overline{r_{2j-2,2k}}   q_{2j-1, 2k}r_{2j-1, 2k}$$
is the total phase accumulated on the scattering events. 

In case all scattering matrices are off-diagonal, i.e. $r_{j,k}=0 \Leftrightarrow  |t_{j,k}|=1$, the subspaces 
$$\bH_{\circlearrowleft}^{j,k}:= \mbox{span}\{|2j,2k\rangle, |2j+1,2k\rangle, |2j+1,2k+1\rangle, |2j,2k+1\rangle \}$$
 are invariant under $U_{CC}$, and the dynamics on those plaquettes is that of left turners with representation in the corresponding basis 
$$
U_{CC}|_{\bH_{\circlearrowleft}^{j,k}}=\begin{pmatrix} 
0 & 0 & 0 & q_{2j-1, 2k}\overline{t_{2j-1, 2k}} \cr
-q_{2j, 2k}t_{2j,2k} & 0 & 0 & 0 \cr
0 & -q_{2j+1, 2k}t_{2j+1,2k}  & 0 & 0 \cr
0 & 0 & q_{2j, 2k+2}\overline{t_{2j,2k+2}}  & 0
\end{pmatrix}
$$
and spectrum $e^{\alpha^L_{2k,2j}}\{1,i, -1, -i\}$, where 
$$e^{4\alpha^L_{2k,2j}}=q_{2j, 2k}t_{2j,2k} q_{2j+1, 2k}t_{2j+1,2k} q_{2j, 2k+2}\overline{t_{2j,2k+2}}  q_{2j-1, 2k}\overline{t_{2j-1, 2k}}$$
is the total phase accumulated on the scattering events. 
\end{lem}

Next we will show that the Chalker-Coddington model is uniformly Lipschitz in its defining scattering matrices, (\ref{eq:S}). Denoting the model defined by $\bbS$ by  $U_{CC}\equiv U_{CC}({\bbS})$ we have:
\begin{lem}\label{contCC}
 There exists a $c>0$ such that for all $\bbS, \bbS'$
 $$
 \|U_{CC}({\bbS})-U_{CC}({\bbS'})\|\leq c\|{\bbS}-{\bbS'}\|_\infty.
 $$
\end{lem}
\begin{proof} For the homeomorphism defined in (\ref{eq:su2}) it holds for the Hilbert Schmidt norm $\Vert M\Vert_{HS}=\left(\tr\ M^*M\right)^{1/2}$:
\[\Vert S(q,r,t)-S(q',r',t')\Vert_{HS}\le \sqrt{2}\left( \vert q-q'\vert + \left(\vert r-r'\vert^2+\vert t-t'\vert^2\right)^{\frac{1}{2}}\right)\] 
and we get
\[ \| U_{CC}({\bbS})-U_{CC}({\bbS'})\|\leq c \left(\sup_{\mu\in\bZ\times2\bZ}\vert q_\mu-q'_\mu\vert + \left(\vert r_\mu-r'_\mu\vert^2+\vert t_\mu-t'_\mu\vert^2\right)^{\frac{1}{2}}\right)\]
{were $c$ comes from  $\sqrt{2}$ and the equivalence of the euclidian matrix norm and the Hilbert Schmidt norm on the scattering matrices.}
\end{proof}
Remark that $\bS^1\times\bS^3$ is  path-connected so any two Chalker-Coddington models can be continuously deformed into one another in the natural topology.

\subsection{Interface between left and right phases}\label{intersec}

We now consider  the situation where the scattering matrices (\ref{eq:S}) describe a left moving phase and a right moving phase separated by a vertical interface  $I_{n_L, n_R}\subset l^2(\bZ^2)$ which turns out to be invariant:

\begin{lem} Let $\bbS$ in (\ref{eq:S}) be such that for $n_L\leq n_R $
\begin{equation}S_{j,2k} \ \mbox{is}\ \left\{\begin{matrix} \mbox{off-diagonal} & \Leftrightarrow & r_{j,2k}=0& \mbox{if} \ j<n_L \cr 
\mbox{diagonal} &\Leftrightarrow & t_{j,2k}=0  & \mbox {if} \ j\geq n_R\end{matrix} \right.\label{eq:chiralphases}\end{equation}
Denote $\evenunder{n_L}$  the largest even integer less or equal to $n_L$ and $\evenover{n_R}$ the smallest even integer greater or equal to $n_R$.  then
\begin{equation}
\label{interface}
I_{n_L, n_R}:=\ell^2\left(\{\evenunder{n_L}, \evenover{n_R}\}\times\bZ; \bC\right)
\end{equation}
is invariant
 under $U_{CC}(\bbS)$. Define
\[U_{I_{n_L,n_R}} {\rm \ the\ restriction\ of\ } U_{CC} {\rm \ to\ } I_{n_L, n_R}.\]
The chiral boundary condition reads for  any  $k\in\bZ$ 
\begin{eqnarray}
U_{I_{n_L,n_R}} \Ket{\evenunder{n_L},2k+1}&=&p \Ket{\evenunder{n_L},2k}\cr 
U_{I_{n_L,n_R}} \Ket{\evenunder{n_R},2k}&=&p \Ket{\evenover{n_L},2k-1}.\label{eq:cbc}
\end{eqnarray}
where  $p\in\bS^1$ is a site dependent phase factor.
In addition it holds for $n_L$ odd:
\[U_{I_{n_L,n_R}} \Ket{\evenunder{n_L},2k}=p\Ket{\evenunder{n_L}+1,2k}\hbox{ and }
U_{I_{n_L,n_R}} \Ket{\evenunder{n_L+1},2k+1}=p\Ket{\evenunder{n_L},2k+1}%
,\]
 and for $n_R$ odd:
\[U_{I_{n_L,n_R}} \Ket{\evenover{n_R},2k+1}=p\Ket{\evenover{n_R}-1,2k+1}\hbox{ and }
U_{I_{n_L,n_R}} \Ket{\evenover{n_R}-1,2k}=p\Ket{\evenover{n_R},2k}
.\]
\end{lem}

\begin{proof}By Lemma \ref{LR} we have that if $n_L=2p_L+1$ then the rightmost invariant plaquettes on the left are $\bH_{\circlearrowleft}^{p_L-1, k}$, if $n_L=2p_L$ : $\bH_{\circlearrowleft}^{p_L-1, k}$; for $n_R=2p_R$ the leftmost invariant plaquettes on the right are $\bH_{\circlearrowright}^{p_R+1,k}$, for $n_R=2p_R-1$: $\bH_{\circlearrowright}^{p_R+1,k}$  \end{proof}

\begin{rems}There is no general restriction to the dynamical behavior of $U_{I_{n_L,n_R}}$. We give some exemples:
\begin{enumerate}\item
A sharp interface $\ell^2\left(\{0\}\times\bZ; \bC\right)$ occurs for $n_L=n_R=0$, and the dynamics are
\begin{align*}
&U_{I_{0,0}} |0,2k+1\rangle = q_{-1,2k} \overline{t_{0,2k}}  |0,2k\rangle\\
&U_{I_{0,0}} |0,2k\rangle = q_{0,2k} r_{0,2k}|0,2k-1\rangle, \ \ \forall k\in \bZ.
\end{align*}
Since all coefficients at the right hand side have modulus one, we deduce that 
$U_{I_{0,0}}$ is unitarily equivalent to the shift on $l^2(\bZ)$ describing a current along the interface.

\item In the case where the strip $\ell^2\left(\{0,1,2\}\times\bZ; \bC\right)$ occurs as $I_{1,1}$
the restriction of $U_{I_{1,1}}$ to the interface 
is also unitarily equivalent to the shift on $l^2(\bZ)$, with a winding snakelike motion, see figure \ref{fig:snake}.
\item Models containing invariant plaquettes in the interface are readily constructed, see figure \ref{fig:stripplaquette}
\end{enumerate}
\end{rems}

\begin{figure}
    \centering
    \begin{subfigure}[b]{0.43\textwidth}
        \includegraphics[width=\textwidth]{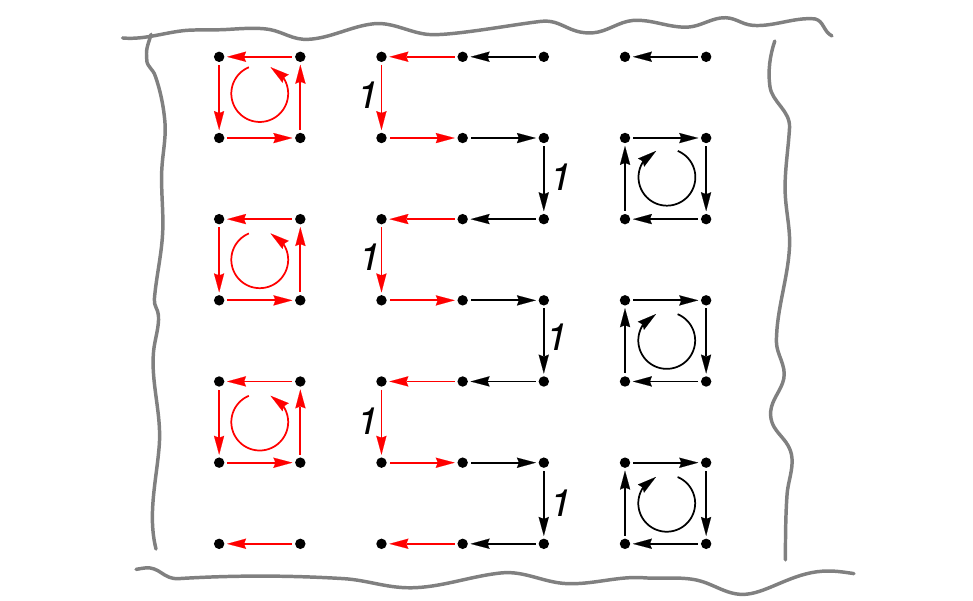}
        \caption{A snake path}
	\label{fig:snake}
    \end{subfigure}\qquad\hspace{1cm} 
    \begin{subfigure}[b]{0.43\textwidth}
        \includegraphics[width=\textwidth]{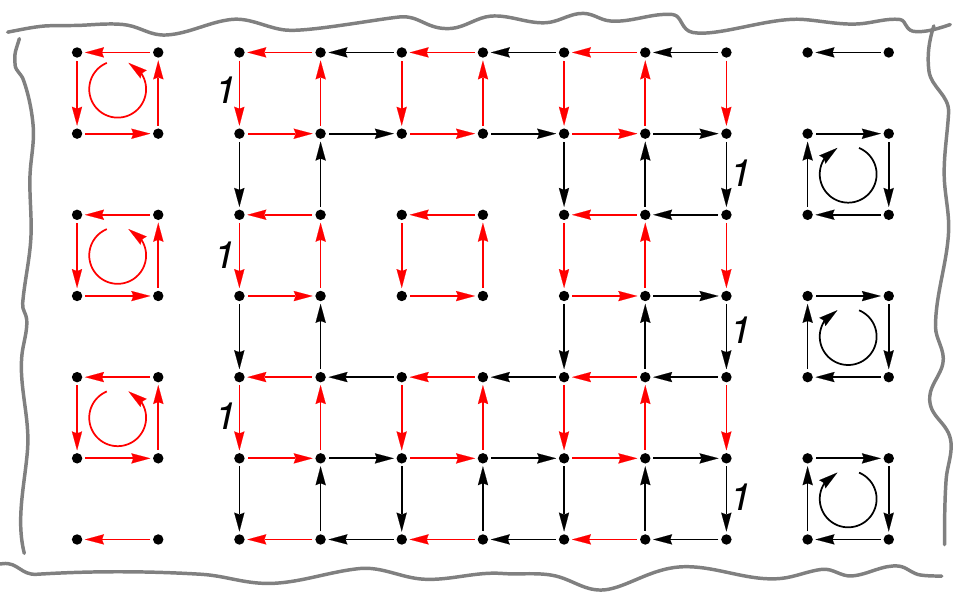}
        \caption{An invariant plaquette}
        \label{fig:stripplaquette}
    \end{subfigure}
\end{figure}

\section{Index and spectrum}
In the following we will use the notations $\sigma_{x}(A)$, $x\in \{pp, cont, ess, ac, sc\}$ to denote the pure point, continuous, essential, absolutely continuous, singular continuous  parts of the spectrum a normal operator $A$.

From a viewpoint of second quantisation
 we shall prove that the chiral boundary condition forces one particle per time step to  flow out of a half strip and deduce that the continuous spectrum is not empty. We shall do so in two different ways: first we use a topological argument,  second we provide an explicit spectral analysis of the relevant flux operator.

In $ I_{n_L,n_R}$ we consider the
projection  $Q$ on the upper half strip
defined as multiplication by $\chi(\lbrack1,\infty))$, the characteristic function of the upper half strip,  i.e.:  $$Q\psi(j,k):=\chi(k\ge1)\psi(j,k).$$ Consider the flux observable
\[\flux:= U_{I}^*QU_{I}-Q.\]
  $\flux$ measures the difference of the number of particles in the half strip at  time one and time zero.
We shall prove that  one particle  per time step is lost from a full  half strip, i.e.:
\begin{thm}\label{thm:main} Let $U_{CC}=U(\bbS)$  be such that
\[r_{j,2k}=0 \mbox{ if } \ j<n_L, \ \  t_{j,2k}=0   \mbox { if } \ j\geq n_R\]
then it holds 
$$
\mbox{\em Tr} (\Phi)=-1.
$$
\end{thm}
We deduce
\begin{cor}\label{sigcont}
$$
\sigma_{cont}(U_{CC})=\sigma_{cont}(U_{I_{n_L,n_R}})\neq \emptyset.
$$
\end{cor}

To prove this result we make use of  the topological invariance of $\tr\flux$,  due to the fact that it is equal to the relative index of two projections, and the spectral flow formula of Kitaev which is very handy for network models. To be self-contained and  to provide some background on our line of thought,  we collect several concepts and results from \cite{ASS,ASS1,ki}. 

\begin{definition}Let $P,Q$ be selfadjoint projections such that $P-Q$ is compact. Their relative index is defined by
\[\ind(P,Q):=\dim\ker(P-Q-1)-\dim\ker(P-Q+1).\]
\end{definition}
The index has the following properties
\begin{thm}\label{indexproperties}
\begin{enumerate}
\item $\ind(P,Q)=\dim\ran P\cap\ker Q-\dim\ran Q\cap\ker P$
\item If $(P-Q)^{2n+1}$ is trace class  for some $n\in\bN_0$ then
\[\ind(P,Q)=\tr(P-Q)^{2n+1}.\]
\item \label{eq:fredholmindexformula}If $P=U^\ast QU$ for a unitary $U$ then $QUQ$ is Fredholm on $\ran Q$ and for its Fredholm index it holds
\[\ind(U^\ast QU,Q)=-ind(QUQ):=\dim\ker QU^\ast Q-\dim\ker QUQ.\]
\item If $U^\ast QU-Q$ is compact then  ${U^\ast}^n QU^n-Q$ is compact for all $n\in\bN$ and
\[\ind({U^\ast}^n QU^n,Q)=n \ind(U^\ast QU,Q).\]
\item If $P(t)=U^\ast(t) QU(t)$ with $\lbrack0,1\rbrack\mapsto U(t)$ is a norm--continuous family of unitary operators then 
\[\ind(P(t),Q)=\ind(P(0),Q)\quad\forall t\]
\item \label{compact}\[\ind(U_1^\ast QU_1,Q)=\ind(U_0^\ast QU_0,Q)\]
for unitaries $U_1,U_0$ such that $U_1-U_0$ is a compact operator.
\end{enumerate}
\end{thm}

\begin{proof} The first four assertions were proven in \cite{ASS,ASS1}. The last two assertions follow the  invariance properties of the Fredholm index of $QUQ$. \end{proof}

\begin{thm}\label{thm:kitaev}In the Hilbert space $\ell^2\left(\bZ; \bC^d\right)$ consider a unitary operator $U$ 
with the  $d\times d$ matrix valued kernel $U(x,y)$  
\[U\psi(x)=\sum_{y\in\bZ}U(x,y)\psi(y)\] 
such that for an {$\alpha>2$}, a positive constant $c$ and all $x\neq y$
\begin{equation}\Vert U(x,y)\Vert_{HS}\le \frac{c}{\vert x-y\vert^\alpha}.\label{eq:decay}\end{equation}
It holds with the half space projection $Q\psi(x):=\chi(x\ge1)\psi(x)$
\begin{enumerate}
\item $U^\ast Q U-Q$ is trace class and
\begin{equation}\ind(U^\ast Q U,Q)=\tr\left(U^\ast Q U-Q\right)=\sum_{z\ge1}\sum_{y<1}\left(\Vert U(z,y)\Vert_{HS}^2-\Vert U(y,z)\Vert_{HS}^2\right)\label{eq:kitaev}\end{equation}
\item\label{kitaevwinding} If  $U(x,y)=M(x-y)$ and (\ref{eq:decay}) is assumed for $\alpha>2$ then  the Fourier transform $\widehat{M}(k):=\sum_{z\in\bZ} M(z) e^{i k z}$  is a periodic matrix valued $C^1$ function  and it holds
\[\ind(U^\ast Q U,Q)=-ind(QUQ)=wind(\det \widehat M)=\frac{1}{2\pi i}\int_0^{2\pi}\tr{\widehat{M}^\ast\partial_k \widehat{M}(k)}\ dk\]
the winding number of the determinant of $\widehat{M}$.
\end{enumerate}
\end{thm}

\begin{proof}1.  We use the notation $<x>:=\left(1+x^2\right)^{1/2}$.

For $\Phi:=U^\ast Q U-Q=U^\ast (Q U-UQ)$ it holds 
\begin{eqnarray*}
&&\Phi(x,y)= \sum_{z}U(z,x)^\ast U(z,y)\left(\chi(z>0)-\chi(y>0)\right)\cr
&=&  \sum_{z}U(z,x)^\ast U(z,y)\left(\chi(z>0)\chi(y\le0)-\chi(z\le0)\chi(y>0)\right)
\end{eqnarray*}
{and thus for a $c>0$
\[\Vert\Phi(x,y)\Vert_{HS}\le c\sum_z \frac{1}{{<z-x>}^\alpha}\frac{1}{{< z-y>}^\alpha}\left(\chi(z>0)\chi(y\le0)+\chi(z\le0)\chi(y>0)\right). \]
Using the estimate
\[\sum_{x} \frac{1}{{< z-x>}^\alpha}\frac{1}{{< z-y>}^\alpha}\le \frac{const}{{<y-z>}^\alpha}\]
we have
\[\sum_{y,z}\sum_x \frac{1}{{< z-x>}^\alpha}\frac{1}{{<z-y>}^\alpha}\left(\chi(z>0)\chi(y\le0)+\chi(z\le0)\chi(y>0)\right)
\]
\[\le \sum_{x\in\bZ}\sum_{y=-x+1}^x\frac{const}{< x>^\alpha}<\infty\]
and thus
\[\sum_{x,y} \Vert\Phi(x,y)\Vert_{HS} <\infty\]
which implies by \cite{bs} }
 that $U^\ast Q U-Q$ is trace class with trace
\[\tr{U^\ast Q U-Q}=\sum_{y\le0}\sum_{z>0}\left(\Vert U(z,y)\Vert_{HS}^2-\Vert U(y,z)\Vert_{HS}^2\right)\]
 By theorem \ref{indexproperties} $\ind(U^\ast Q U,Q)=\tr(U^\ast Q U-Q)$ and
(\ref{eq:kitaev}) follows.

2. The first equality is 1. together with Theorem \ref{indexproperties}.\ref{eq:fredholmindexformula}. As for the second equality observe 
\[wind(det\widehat{M})=\frac{1}{2\pi i}\int_0^{2\pi}\tr{\widehat{M}^\ast\partial_k \widehat{M}(k)}\ dk=\sum_{z\in\bZ} z \tr M(z)^\ast M(z)\]
\[=\sum_{z>0}\left( z \tr M(z)^\ast M(z)- z \tr M(-z)^\ast M(-z)\right)=\tr{U^\ast Q U-Q}.\] 
\end{proof}

\begin{rems} \label{cutindependence}
\begin{enumerate}
\item The proof of theorem \ref{thm:kitaev}.\ref{kitaevwinding} provides (via Theorem \ref{indexproperties}.\ref{eq:fredholmindexformula}) a very explicit proof of the index theorem for Toeplitz matrix operators, i.e.:
\[\ind(\widehat{Q}\widehat{M}\widehat{Q})=-wind(det\widehat{M})\]
for the Fredholm index of the the operator of multiplication by $\widehat{M}$ on the Hardy space $\{f\in L^2\left(\bS^1;\bC^d\right);\check{f}(n)=0,  \forall n\le 0\}$. See,  for example,  \cite{atiyah} for a topological proof. The explicit proof is well known in the scalar case, see \cite{brezis}.
\item    \label{2} The index is independent of the cut position, i.e.: $\hbox{for all } y_0\in\bZ$ 
\[\ind\left(U^\ast\chi(\lbrack y_0,\infty))U,\chi(\lbrack y_0,\infty))\right)=\ind\left(U^\ast\chi(\lbrack1,\infty))U,\chi(\lbrack1,\infty))\right). 
\]
\end{enumerate}
\end{rems}

\begin{proof}$D:=\chi(\lbrack y_0,\infty))-\chi(\lbrack1,\infty)))$ is a finite rank projector so $\tr\left(U^\ast D U-D\right)=\tr U^\ast D U-\tr D=0$\end{proof}

After recalling the above background material we state the proof
\begin{proof}(of {\bf Theorem \ref{thm:main}}) To prove the theorem we use the stability of the index, Theorem \ref{indexproperties}.\ref{compact}. Let $\bbS$ be the $U(2)$ valued map defining $U_{CC}$. Let $U'$ be the unitary defined by $\bbS'$ with $S'_{j,0}:=\begin{pmatrix} 1&0 \cr0&1\end{pmatrix}$ for all odd $j$ between the strip boundaries, $S_{j,2k}=S'_{j,2k}$ elsewhere,  see figure (\ref{fig:stripchanged}). The index corresponding to $U'$ is unchanged because the modification to $U'_I-U_I$ is of finite rank. In Kitaev's formula (\ref{eq:kitaev}) the only non trivial  matrix element left is the one of  $U'_I(0,1)$ on the left boundary of the strip whose modulus equals $1$, thus
\[\ind(U_I^* Q U_I,Q) =\Bra{\evenunder{n_L},1}U_I^* Q U_I-Q\Ket{ \evenunder{n_L},1}=-1.\]
\end{proof}
 
Corollary \ref{sigcont} now follows from the following sufficient condition for delocalisation:

\begin{prop} Let $U$ be a unitary operator on   $\ell^2\left(\bZ; \bC^d\right)$  such that for a $\alpha>2, c>0$
\[U\psi(x)=\sum_{y\in\bZ}U(x,y)\psi(y) \hbox{ with } \Vert U(x,y)\Vert_{HS}\le \frac{c}{\vert x-y\vert^\alpha}\qquad  \forall x\neq y\] 
then
\[\tr\left(U^\ast \chi(\lbrack1,\infty) U -\chi(\lbrack1,\infty)\right)\neq0 \Longrightarrow \sigma_{cont}(U)\neq\emptyset.\]
\end{prop}

\begin{proof}
Denote $P_{pp} ,P_{cont} $ the projections on the pure point and continuous subspace of $U$.  
 For any eigenvector $\ffi$ of $U$, we have
$
\left\langle \ffi, (U^\ast Q U -Q )\ffi \right\rangle = 0,
$
so that 
\[\tr (P_{pp} \Phi P_{pp} )=0.\]
 Hence, by cyclicity of the trace and $P_{cont} P_{pp} =0$,
\begin{align*}
\tr (\Phi)&=\tr ((P_{pp} +P_{cont} )\Phi((P_{pp} +P_{cont} )) )
=\tr (P_{cont} \Phi P_{cont} )\neq0.
\end{align*}
\end{proof}

\begin{figure}[hbt]
\centerline {
\includegraphics[height=7cm]{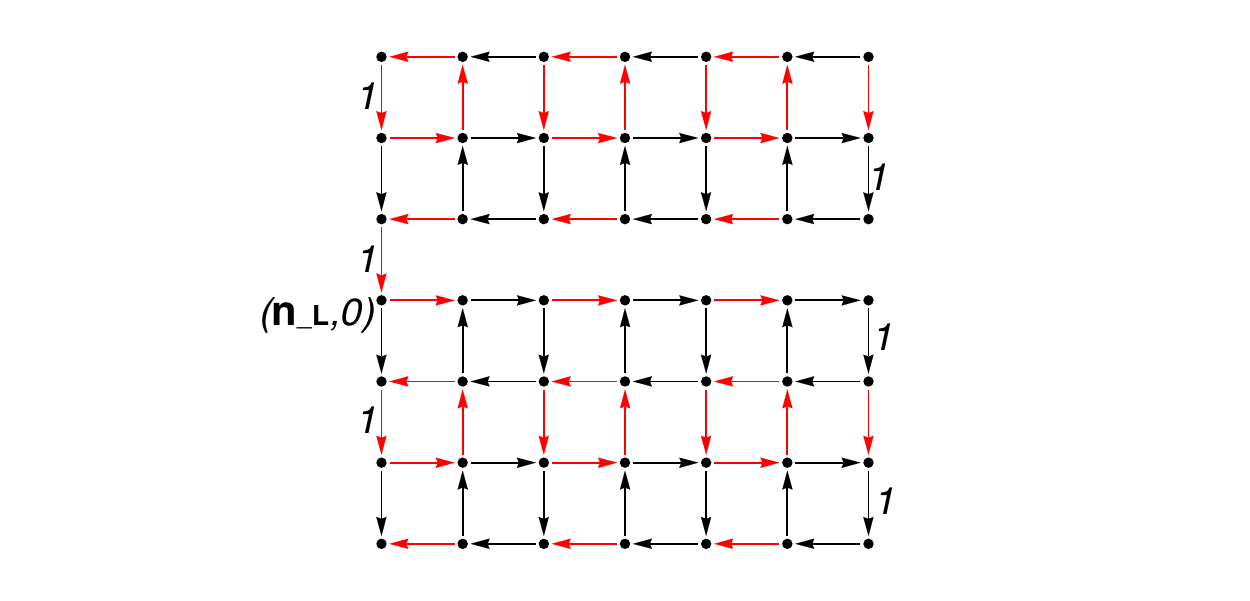}
}
\caption{A finite rank perturbation of $U_I$}
\label{fig:stripchanged}
\end{figure}

The propagation induced by the chiral boundary condition being non-trivial, it is instructive to study the spectrum of the flux observable $\flux$. 

By remark \ref{cutindependence}.\ref{2} the flux through the horizontal interface is actually independent of the cut position. Consider the flux observable through the cut at height $c\in\bZ$ i.e.:

\[\flux_c:= U_{I}^*\chi(\lbrack c,\infty))U_{I}-\chi(\lbrack c,\infty)) {\rm \ on \ } I_{n_L,n_R}.\]

An explicit computation yields the following
\begin{lem} The  finite rank self-adjoint operator $\Phi_c$ takes the following forms, depending on the parameters.\\
If $n_L=2p_L$, $n_R=2p_R$, and $c$ is even,

{\begin{align*}
\Phi_c=-|2p_R,c\rangle\langle 2p_R, c|&+ \bigoplus_{p_l\leq j < p_R} 
\begin{pmatrix}
-|r_{2j,c}|^2 & -\overline{r_{2j,c}t_{2j,c}} \cr - r_{2j,c} t_{2j,c} & |r_{2j,c}|^2
\end{pmatrix},
\end{align*}
}
where the matrices are expressed in the ordered basis $\{|2j,c\rangle, |2j+1,c-1\rangle\}$.\\
If $c$ is odd, 
{\begin{align*}
\Phi_c=-|2p_L,c\rangle\langle 2p_L, c|&+ \bigoplus_{p_L <j \leq p_R}
\begin{pmatrix}
|t_{2j-1,c-1}|^2 & -r_{2j-1,c-1}t_{2j-1,c-1} \cr -\overline{r_{2j-1,c-1}t_{2j-1,c-1}}  & -|t_{2j-1,c-1}|^2
\end{pmatrix},
\end{align*}
}
where the matrices are expressed in the ordered basis $\{|2j-1,c-1\rangle, |2j,c\rangle \}$.\\
In case $n_L=2p_L+1$ and/or $n_R=2p_R-1$, the formulae above hold true with $S_{2p_L,2k}$ off-diagonal and/or $S_{2p_R-1,2k}$ diagonal, forall $k\in \bZ$.
\end{lem}
As a  consequence, we get
\begin{thm}
With the notations above, 
$$
\sigma(\Phi_c)=\left\{\begin{array}{ll}
\{-1\}\cup_{p_L\leq j < p_R} \{+|r_{2j,c}|, -|r_{2j,c}|\}, & c\ \mbox{even}\\
\{-1\}\cup_{p_L< j \leq p_R} \{+|t_{2j-1,c-1}|, -|t_{2j-1,c-1}|\}, & c\ \mbox{odd,}
\end{array}\right.
$$
and 
$$
\mbox{\em Tr} (\Phi_c)=-1, \ \forall c\in\bZ. 
$$
\end{thm}
\begin{rem} One sees that if 
$\lim_{j\to\infty}t_{j,c}= 0$ 
sufficiently fast for some $c$ odd and $n_R\to \infty$, then $\Phi_c$ is trace class and $\sigma_{cont}(U_{CC})\neq\emptyset$. A similar statement holds for $c$ even. \end{rem}

\subsection{The shift and the absolutely continuous spectrum}

To go beyond Corollary \ref{sigcont}, we exploit the stability of the Fredholm index $\tr (\Phi_c)$ and that of  the absolutely continuous spectrum under finite rank perturbations. The idea is to unravel the existence of a shift within $U_I$, modulo finite rank perturbations.
 
\begin{thm}\label{acs} Let $U_{CC}=U(\bbS)$  be such that
\[r_{j,2k}=0 \mbox{ if } \ j<n_L, \ \  t_{j,2k}=0   \mbox { if } \ j\geq n_R\]
There exists a decomposition of $I_{n_L, n_R}$ into orthogonal closed subspaces $\hil_u\oplus \hil_s$ such that $U_I=V\oplus S + F$, where
$V$ is unitary, $S$ is a bilateral shift, and $F$ is a finite rank operator. \\
Consequently,
$$
\sigma_{ac}(U_{CC})= {\mathbb S}^1.
$$
\end{thm}
\begin{rem}
The point spectrum of $U_{CC}$ contains all eigenvalues of the restrictions to individual plaquettes in the left and right moving phases, as described in Lemma \ref{LR}. Other finite dimensional invariant subspaces may also occur, and contribute to the point spectrum and singular continuous spectrum may be present as well. 
\end{rem}
\begin{proof} 
Consider the projection on the upper half plane and its complement
\[Q=\chi\left(\lbrack1,\infty)\right), \quad Q^\perp=\bI-Q.\]
Let $U=U(\bbS)$ and $U'$ the same modification of $U$ as in theorem \ref{thm:main} ,  see figure (\ref{fig:stripchanged}), i.e.:
$U'=U(\bbS')$  with $S'_{j,0}:=\begin{pmatrix} 1&0 \cr0&1\end{pmatrix}$ for all odd $j$ between the strip boundaries,  $S_{j,2k}=S'_{j,2k}$ elsewhere. Then $U'$ is unitary and $U'_I-U_I=:F$ is a finite rank operator so $\sigma_{ac}(U)=\sigma_{ac}(U')$. It holds for a phase $p\in\bS^1$:
\[Q^\perp U'_I Q=p \Ket{\evenunder{n_L},0}\Bra{\evenover{n_L},1}, \quad Q U'_I Q^\perp=0.\]
Consequently 
\[\cL:=span\left\{\Ket{\evenunder{n_L},1}\right\}\] is a wandering subspace in the sense of \cite{SNF}, i.e.:
\[{U'_I}^n \cL\perp \cL\qquad \forall n\ge1.\]
By unitarity if follows
\[{U'_I}^n \cL\perp \cL\qquad \forall n\in\bZ\setminus{0}\]
and that
\[\hil_s:=\bigoplus_{n\in\bZ}{U'_I}^n \cL\]
reduces $U'_I$ which on $\hil_s$ is a bilateral shift; cf. : \cite{SNF},ch.2. So $\sigma(S)=\sigma_{ac}(S)=\bS^1$.\end{proof}

\begin{rem} The idea to consider a wandering subspace was used by von Neumann in his discussion of symmetric operators with non equal defect  indices \cite{vonNeumann}; the idea to use this to get  information on the absolutely continuous spectrum is inspired by \cite{cetal}.
\end{rem}

\section{Translation Invariant Case}

We can refine the spectral analysis of the CC model provided it possesses more symmetries. We assume in this section that the CC model is invariant under vertical translations. In other words, the scattering matrices, see (\ref{eq:S}), are identical on all scattering centers with equal horizontal components:
\be\label{trainv}
S_{j,2k}=S_j, \  \ \mbox{for all} \ (j,k)\in\bZ^2.
\ee

Firstly, for the case of the interface between different chiral phases,  we shall show in addition  that $\sigma_{sc}(U_I)=\emptyset$ and provide an independent proof that $\sigma_{ac}(U_I)=\bS^1$. Secondly we characterise all  vertically translation invariant Chalker Coddington models on $\bZ^2$  explicitly as one dimensional  Quantum Walks.
\subsection{Translation invariant interface}
In this section we assume translation invariance, equation (\ref{trainv}), and different chiral phases, equation (\ref{eq:chiralphases}).

The periodicity of $\bbS$ implies a period-$2$ periodicity of $U_I$. To exploit this we regroup two horizontal slices in a vector and use Fourier transform; define for $\psi\in I_{n_L,n_R}$
\[{\bm \psi}(k)_j:=\left(\psi(j,(2k-1)), \psi(j,2k)\right);\]
for $d:=\evenover{n_R}-\evenunder{n_L}+1$ the corresponding map
\begin{equation}\label{fatmap}
I_{n_L,n_R}\to\ell^2\left(\bZ;\bC^d\right)\otimes\bC^2
\end{equation}
is unitary and we denote by
$\fatU_I$ the operator on $\ell^2\left(\bZ;\bC^d\right)\otimes\bC^2$ corresponding to $U_I$. By periodicity of $\bbS$ we have
\begin{equation}\label{eq:a}
\fatU_I\fatpsi(k)=\sum_{k'\in\bZ} V(k-k')\fatpsi(k')
\end{equation}

where the kernel $k\mapsto V(k)\in {\mathcal B}(\bC^d\otimes \bC^2)$  has compact support. In particular its Fourier transform

\begin{equation}\label{eq:b}
y\mapsto \widehat{V}(y):=\sum_{k\in\bZ} e^{-i y k}V(k)
\end{equation}

is a trigonometric polynomial. 

Remark that in the following we shall not distinguish between $\bS^1$,  the coordinates $\lbrack0,2\pi)\ni y\mapsto e^{-iy}\in\bS^1$ or $\bR/{2\pi\bZ}$.
For its winding number of $\det{\widehat{V}}$ it holds

\begin{prop}For $\widehat{V}$ defined in (\ref{eq:a}), (\ref{eq:b}) it holds
\[ wind\left(\det{\widehat{V}}\right)=-1.\]

\end{prop}

\begin{proof}With $\fatchi(\lbrack1,\infty)):=\chi(\lbrack1,\infty))\otimes\bI$ it holds by Theorem \ref{thm:kitaev}
\[wind\left(\det{\widehat{V}}\right)=\tr\left({\fatU_I^\ast \fatchi(\lbrack1,\infty))\fatU_I-\fatchi(\lbrack1,\infty)})\right)\] On the other hand by unitary equivalence and  Theorem \ref{thm:main} we have
\[\tr\left({\fatU_I^\ast \fatchi(\lbrack1,\infty))\fatU_I-\fatchi(\lbrack1,\infty)})\right)=\tr\left({U_I^\ast \chi(\lbrack1,\infty))U_I-\chi(\lbrack1,\infty)})\right)=-1.\]

\end{proof}

A non trivial winding of the determinant of a continuous matrix-valued multiplication operator implies that its spectrum  is the whole circle. If it is analytic the absolutely continuous spectrum is the whole circle and the singular continuous spectrum is empty. The result is non trivial because eigenvalues of periodic unitaries need not be periodic as the example  $\begin{pmatrix}0&1\cr e^{iy}&0\end{pmatrix}$ shows.

\begin{lem}\label{lemma:goodshit}For $d\in\bN$ consider  $V\in C^0\left(\bS^1; U(d)\right)$,  a continuous unitary matrix-valued function such that
\[wind\left( \det V\right)\neq0.\]
Then it holds
\begin{enumerate}
\item  For all $\lambda\in\bS^1$ there exists a $z\in\bS^1$ such that $\lambda\in\sigma(V(z))$;
\item  for $V$ the operator of multiplication by $V(z)$ on $\ell^2\left(\bS^1;\bC^d\right)$:
\[\sigma(V)=\bS^1.\]
\end{enumerate}
\end{lem}

\begin{proof} 1. 
{If there existed a  $\lambda\in S^1$ such that $\lambda\notin \sigma(V(z)$  for all $z \in S^1$,  the logarithm with cut at $\lambda$ would be well-defined on $\sigma(V (z))$ and thus $\log_\lambda V(z)$ for all $z$;  this would imply
\[wind(det V)=\frac{1}{2\pi i}\oint\frac{d \det(V)}{\det V}=\frac{1}{2\pi i}\oint d\tr\log_\lambda V=0\] 
}

2. $\sigma(V(z))$ depends continuously on $z$ thus for $\lambda\in\bS^1$ the Lebesgue measure 
$\mu_L\left(\sigma(V(.))\cap(\lambda-\varepsilon, \lambda+\varepsilon)\neq\emptyset\right)\neq0\quad \forall\varepsilon>0$ thus $\lambda\in\sigma(V)$.
\end{proof}

It follows

\begin{thm} Under the assumption in equations (\ref{trainv}) and (\ref{eq:chiralphases}) it holds
$$
\sigma_{sc}(U_{CC})= \emptyset \hbox{ and } \sigma_{ac}(U_{CC})=\bS^1.
$$
\end{thm}
\begin{proof}
It is sufficient to show the assertions for $U_I$ as the  spectrum outside the strip $I$ is explicit.  

By unitary equivalence the spectrum is the same as the one of the multiplication operator by $\widehat{V}$ on $L^2(\bS^1; \bC^d)$. The Fourier transform $\widehat{V}$ is a trigonometric polynomial thus analytic;  so its eigenvalues can be chosen to be analytic \cite{k}. Thus the spectrum of the fibered operator consists of the ranges of a finite set of eigenvalues made of absolutely continuous spectrum, see \cite{rs}, Theorem XIII.86,  which by Lemma \ref{lemma:goodshit}, is the whole circle.

\end{proof}

\subsection{General translation invariant model and reduction to a family of one dimensional  Quantum Walks}
Using a more carefully chosen partial Fourier transformation we now show that a vertically translation invariant Chalker Coddington model can be represented as a direct integral of a family of one dimensional  Quantum Walks.

\begin{prop} Consider a Chalker Coddington Model $U_{CC}=U_{CC}(\bbS)$, assume periodicity for $\bbS$ as in equation (\ref{trainv}). On
$\ell^2(\bZ)\otimes\bC^2$ define the unitary Quantum Walk
\[U_{QW}(y):=\fatS C(y)\]
where with $\Ket{+}=(1,0), \Ket{-}=(0,1)$,
\[\fatS \Ket{j}\otimes\Ket{+}:=\Ket{j+1}\otimes\Ket{+}, \quad \fatS \Ket{j}\otimes\Ket{-}:=\Ket{j-1}\otimes\Ket{-},\]
\[C(y)\ket{j}\otimes v:=\ket{j}\otimes C_j(y)v\qquad \forall v\in\bC^2\]
\[C_{2j}(y):=q_{2j}\begin{pmatrix}-t_{2j}& \overline{r_{2j}}e^{iy} \cr r_{2j}e^{-iy}&\overline{t_{2j}}\end{pmatrix},\quad
C_{2j+1}(y):=q_{2j+1}\begin{pmatrix}r_{2j+1}&\overline{t_{2j+1}} \cr -t_{2j+1}&\overline{r_{2j+1}}\end{pmatrix}.\]
Then $U_{CC}$ is unitarily equivalent to the fibered operator
\[\int_{\bS^1}^\oplus U_{QW}(y)\ \frac{dy}{2\pi} \quad\hbox{ on } L^2\left(\bS^1;\ell^2(\bZ)\otimes\bC^2\right).\]
\end{prop}
\begin{proof}

Define the unitary
\[\cF:\ell^2\left(\bZ^2; \bC\right)\to L^2(\bS^1)\otimes\ell^2(\bZ;\bC)\otimes\bC^2\]
by
\[\cF\Ket{j,2k}:=e^{iky}\otimes\Ket{j}\otimes\Ket{+}=e^{iky}\otimes\Ket{j,+},\quad \cF\Ket{j,2k+1}:=e^{iky}\otimes\Ket{j-1}\otimes\Ket{-}=e^{iky}\otimes\Ket{j-1,-}.\] 
By the definition of $U_{CC}(\bbS)$, equations (\ref{def:UCC}) we have

\begin{align*}
&\begin{pmatrix} \cF U_{CC}\cF^{-1} e^{iky}\otimes\Ket{2j,+} \cr 
\cF U_{CC}\cF^{-1} e^{-iy}e^{iky}\otimes\Ket{2j,-}  \end{pmatrix} = S_{2j} \begin{pmatrix} e^{-iy}e^{iky}\otimes\Ket{2j-1,-} \cr e^{iky}\otimes\Ket{2j+1,+} \end{pmatrix},\notag \\
&\begin{pmatrix} \cF U_{CC}\cF^{-1} e^{iky}\otimes\Ket{2j+1,+} \cr 
\cF U_{CC}\cF^{-1}e^{iky}\otimes\Ket{2j+1,-} \end{pmatrix} = S_{2j+1} \begin{pmatrix} e^{iky}\otimes\Ket{2j+2,+}  \cr e^{iky}\otimes\Ket{2j,-} \end{pmatrix}.
\end{align*}
So $C_{2j+1}=S^T_{2j+1}$  and $C_{2j}=\begin{pmatrix}0&1\cr e^{-iy}& 0\end{pmatrix}S^T_{2j}\begin{pmatrix}1&0\cr0&e^{iy}\end{pmatrix}$ which gives the result.
\end{proof}

\begin{prop}\label{defMQW}Consider a Chalker Coddington Model $U_{CC}=U_{CC}(\bbS)$, assume periodicity for $\bbS$ as in equation (\ref{trainv}). For $y\in\lbrack0,2\pi)$ consider the unitary operator defined on $\ell^2(\bZ)$ by the  matrix $M_{QW}(y):=$

     \begin{equation}\label{matMQW}
     \mkern-60mu\bordermatrix{
      & &  &      4j   \downarrow           &                            &    &    &   \cr
      & \ddots &  0  &                    &                          &    &    &    &   \cr
         &  & 0 & e^{-iy}r_{2j}q_{2j} & \overline{t_{2j}}q_{2j} &    &    &    &   \cr 
       4j_\rightarrow &  & \overline{t_{2j-1}}q_{2j-1} & 0 & 0 &    &    &    &   \cr 
          &  &  & 0 & 0 &  -t_{2j+1}q_{2j+1}  &  \overline{r_{2j+1}}q_{2j+1}  &    &   \cr 
         &   &  & -t_{2j}q_{2j} & e^{iy}\overline{r_{2j}}q_{2j} & 0  & 0  &    &   \cr 
         &   &     &           &                                  &  0  & 0  &  e^{-iy}r_{2j+2}q_{2j+2}  &   \cr 
          &   &  &  &  &  r_{2j+1}q_{2j+1}  &  \overline{t_{2j+1}}q_{2j+1}  & 0   &   \cr 
&& & & & & & 0 & \ddots
}.
\end{equation}
Then $U_{CC}$ is unitarily equivalent to 
\[\int_{\bS^1}^\oplus M_{QW}(y)\ \frac{dy}{2\pi} \qquad \hbox{ on } L^2\left(\bS^1;\ell^2(\bZ)\right).\]
\end{prop}
\begin{proof}This representation follows in a general way from the representation as Quantum Walk, see \cite{ABJ3}.
Explicitely, define the unitary
\[\cG:\ell^2\left(\bZ^2; \bC\right)\to L^2(\bS^1)\otimes\ell^2(\bZ;\bC)\]
by
\begin{equation}\label{relabel}
\cG\Ket{j,2k}:=e^{iky}\otimes\Ket{2j},\quad \cG\Ket{j,2k+1}:=e^{iky}\otimes\Ket{2j-1}.
\end{equation}

By the definition of $U_{CC}(\bbS)$, equations (\ref{def:UCC}) we have

\begin{align*}
&\begin{pmatrix} \cG U_{CC}\cG^{-1}e^{iky}\otimes\Ket{4j} \cr \cG U_{CC}\cG^{-1}e^{-iy}e^{iky}\otimes\Ket{4j+1}  \end{pmatrix} = S_{2j} \begin{pmatrix} e^{-iy}e^{iky}\otimes\Ket{4j-1} \cr e^{iky}\otimes\Ket{4j+2} \end{pmatrix},\notag \\
&\begin{pmatrix} \cG U_{CC}\cG^{-1}e^{iky}\otimes\Ket{4j+2} \cr \cG U_{CC}\cG^{-1} e^{iky}\otimes\Ket{4j+3} \end{pmatrix} = S_{2j+1} \begin{pmatrix} e^{iky}\otimes\Ket{4j+4} \cr e^{iky}\otimes\Ket{4j+1}  \end{pmatrix}
\end{align*}
which gives the result.
\end{proof}

\begin{rem}\label{phasem} Up to a unitary equivalence by a $y$ independent operator, we can assume the matrix elements of the operator $M_{QW}(y)$ satisfy
$$r_{2j+1}=|r_{2j+1}|, \ t_{2j}=i|t_{2j}| , \ \forall \ j\in\bZ.$$ 
\end{rem}
\begin{proof} Apply Lemma 3.2 in \cite{bhj}, which shows that we can choose the phases of those elements of the infinite five-diagonal unitary matrix at hand, modulo explicit unitary equivalence by a diagonal operator. \end{proof}

\begin{rem} Proposition \ref{defMQW} allows us to provide a direct  proof of the identity $wind (\det (M_{I}))=-1$, where $M_I(y)$ is the Fourier transform representation of the restriction $U_I$ corresponding to (\ref{matMQW}).
\end{rem}
\begin{proof}
The Fourier image $M_I(y)$ is a finite unitary matrix of the form (\ref{matMQW}), taking into account the finite width of the interface (\ref{interface}). Due to the relabelling (\ref{relabel}), and with $\evenover{n_R}=2p_R$, $\evenover{n_L}=2p_L$, the matrix $M_I(y)$ acts on 
$$l^2(\{4p_L-1, 4p_L, \cdots, 4p_R-1, 4p_R\})\simeq \bC^{2(2(p_R-p_L)+1)},$$
where $S_{2p_L-1}$ is off-diagonal and $S_{2p_R}$ is diagonal. Hence,
\begin{align*}
&M_I(y) |4p_L-1\rangle = \overline{t_{2p_L-1}}q_{2p_L-1}|4p_L\rangle\\
&M_I(y) |4p_R\rangle = e^{-iy}\overline{r_{2p_R}}q_{2p_L}|4p_R-1\rangle,
\end{align*}
and on the remaining $4(p_R-p_L)$ vectors, $M_I(y)$ has the block structure depicted in (\ref{matMQW}).
Lemma 3.2 in \cite{bhj} again allows us replace $t_{2j}$ by $e^{iy}t_{2j}$ in the matrix elements 
of $M_I(y)$, up to a unitary transform that depends (analytically) on $y$. Thus, for our spectral considerations, we can assume without loss that all rows carry the same phase factors, so that 
$$
M_I(y)=D(y)M_I(0), \ \mbox{where} \ D(y)=\mbox{\rm diag}(e^{-iy}, 1, 1, e^{iy}, e^{-iy}, 1 \cdots, 1, e^{iy}, e^{-iy}, 1 ).
$$
Due to the excess of phase $e^{-iy}$, the winding number of the $2\pi$-periodic analytic map $y\mapsto \det (M_I(y))=e^{-iy}\det (M_I(0))$ equals $-1$. \end{proof}

\newpage


\end{document}